\documentclass[a4paper, 12pt]{article}
\usepackage{cmap} 
\usepackage[T2A]{fontenc}
\usepackage[utf8x]{inputenc}
\usepackage[russian, english]{babel}
\usepackage{amsmath,amssymb}

\usepackage{amsthm}

\usepackage{yfonts}
\usepackage{hyperref}
\usepackage{graphicx}
\usepackage{IEEEtrantools}
\usepackage{multirow}

\usepackage{ dsfont }

\theoremstyle{plain}

\newtheorem{thm}{Theorem}

\newtheorem{lemma}{Lemma}
\usepackage{subfig}

\usepackage[left=3cm, right=3 cm]{geometry}

\begin{document}

\title{The exact complexity of a Robinson tiling}
\date{}
\author{Galanov Ilya \\ \href{mailto:galanov@lipn.univ-paris13.fr}{galanov@lipn.univ-paris13.fr}}
\maketitle

\begin{abstract}

We find the exact formula for the number of distinct $n \times n$ square patterns which appear in a Robinson tiling made of one infinite order supertile.

\end{abstract}

\section{Introduction}
Raphael Robinson in his work about the undecidability of the domino problem  $[\ref{RR}]$ introduced a set of six tiles depicted in Figure  ~\ref{robinson-tiles}. These tiles can be rotated and reflected, one tile can fit the other only in such a way that the arrowhead matches the arrow tail and each $2 \times 2$ block must contain exactly one \emph{bumpy corner}, the leftmost in Figure 1.

\begin{figure}[htp]
 \begin{center}
  \subfloat[]{\includegraphics[width=0.15\textwidth]{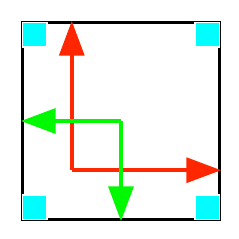}\label{fig:f1} }
  \subfloat[]{\includegraphics[width=0.15\textwidth]{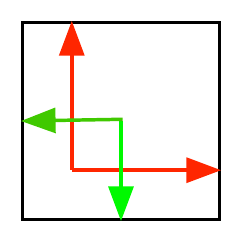}\label{fig:f2}} 
  \subfloat[]{\includegraphics[width=0.15\textwidth]{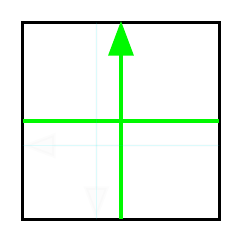}}   
  \subfloat[]{\includegraphics[width=0.15\textwidth]{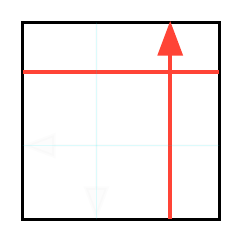}}   
  \subfloat[]{\includegraphics[width=0.15\textwidth]{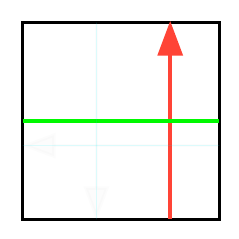}}   
  \subfloat[]{\includegraphics[width=0.15\textwidth]{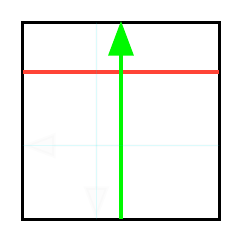}}   
  \end{center}
  \caption{Tiles of type (a) are called \emph{bumpy corners}, tiles of type (b) are called \emph{corners}, all the other tiles are called \emph{arms}.}
\label{robinson-tiles}
\end{figure}

 \begin{figure}[!tbp]
  \centering
  \subfloat[]{\includegraphics[width=0.213\textwidth]{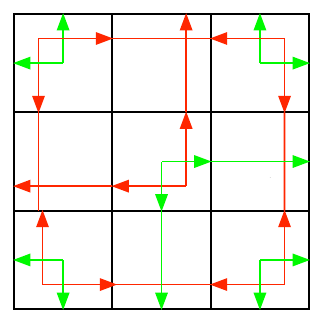}\label{fig:f1} }
  \qquad
  \qquad
  \subfloat[]{\includegraphics[width=0.4\textwidth]{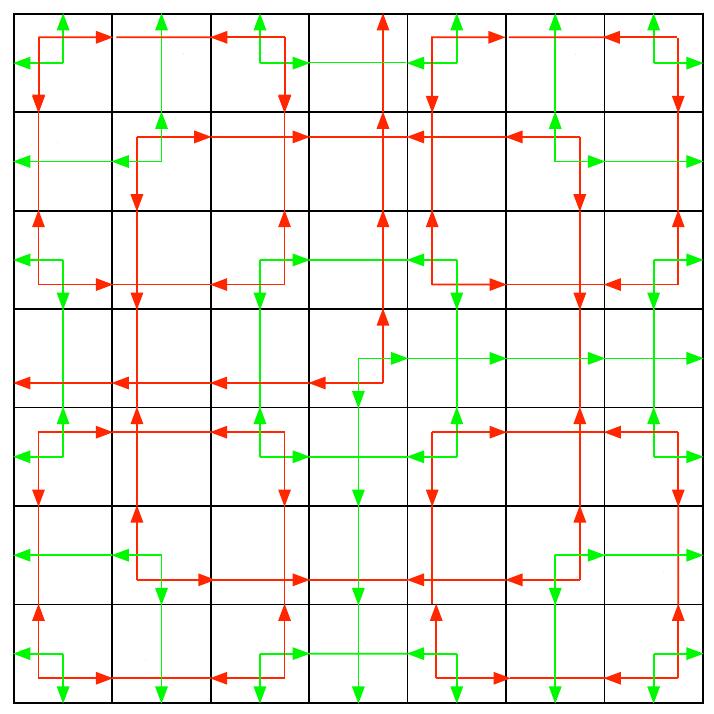}\label{fig:f2}}
  \caption{Supertiles of the second and third rank. }
\end{figure}
It is possible to tile the Euclidean plane with copies of these six tiles, but only in an \emph{aperiodic} way. The key to this result is that any Robinson tiling has a hierarchical structure: for all $n$, define a supertile of rank $n$ as shown in Figure 2.
Bumpy corner tiles are said to be supertiles of the first rank. An increasing union of supertiles of rank $n$ is called an infinite order supertile.  The Robinson tiling can be either made of only \emph{one} infinite order supertile or contain \emph{two} or  \emph{four} infinite order supertiles.

We will prove that the number of distinct blocks of size $n \times n$ (with $n \ge 2$) that appear in Robinson tiling made of one infinite order supertile is defined by the formula 

\begin{IEEEeqnarray}{cccl}
\label{comp}
A(n) & =  &  \> 32n^{2}  + 72 n \cdot 2^{\lfloor\log_2 n\rfloor} - 48 \cdot 2^{2\lfloor\log_2 n\rfloor} .
\end{IEEEeqnarray}

Similar result for the number of factors in \emph{paper folding sequence} has been obtained by Jean-Paul Allouche in [\ref{allouche}]. 

\section{Complexity of Robinson tiling}

In this paper, we say that a tiling of an $n \times n$ square is $\emph{correct}$ if it can appear in a Robinson tiling made of one infinite supertile. In this case, defining a correct $n \times n$ block is equivalent to defining the hierarchy of intersecting squares, see Figures 3 and 4. This can be done by placing corner tiles corresponding to supertiles of all ranks, note that some of them may be outside of our block.

\begin{figure}[!tbp]
  \label{hier-a}
  \centering
  \subfloat[]{\includegraphics[width=0.4\textwidth]{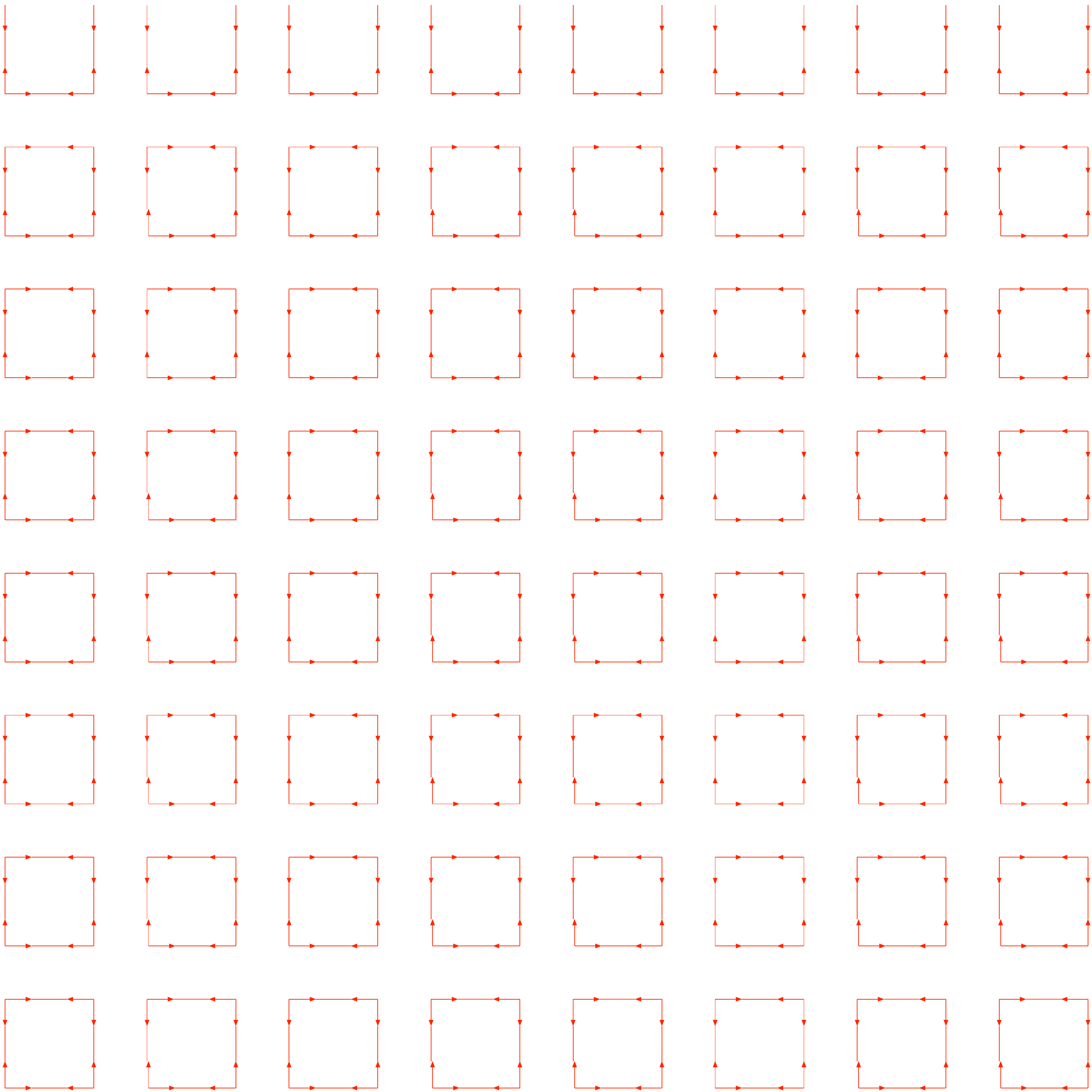}\label{fig:f1} }
  \qquad
  \qquad
  \subfloat[]{\includegraphics[width=0.4\textwidth]{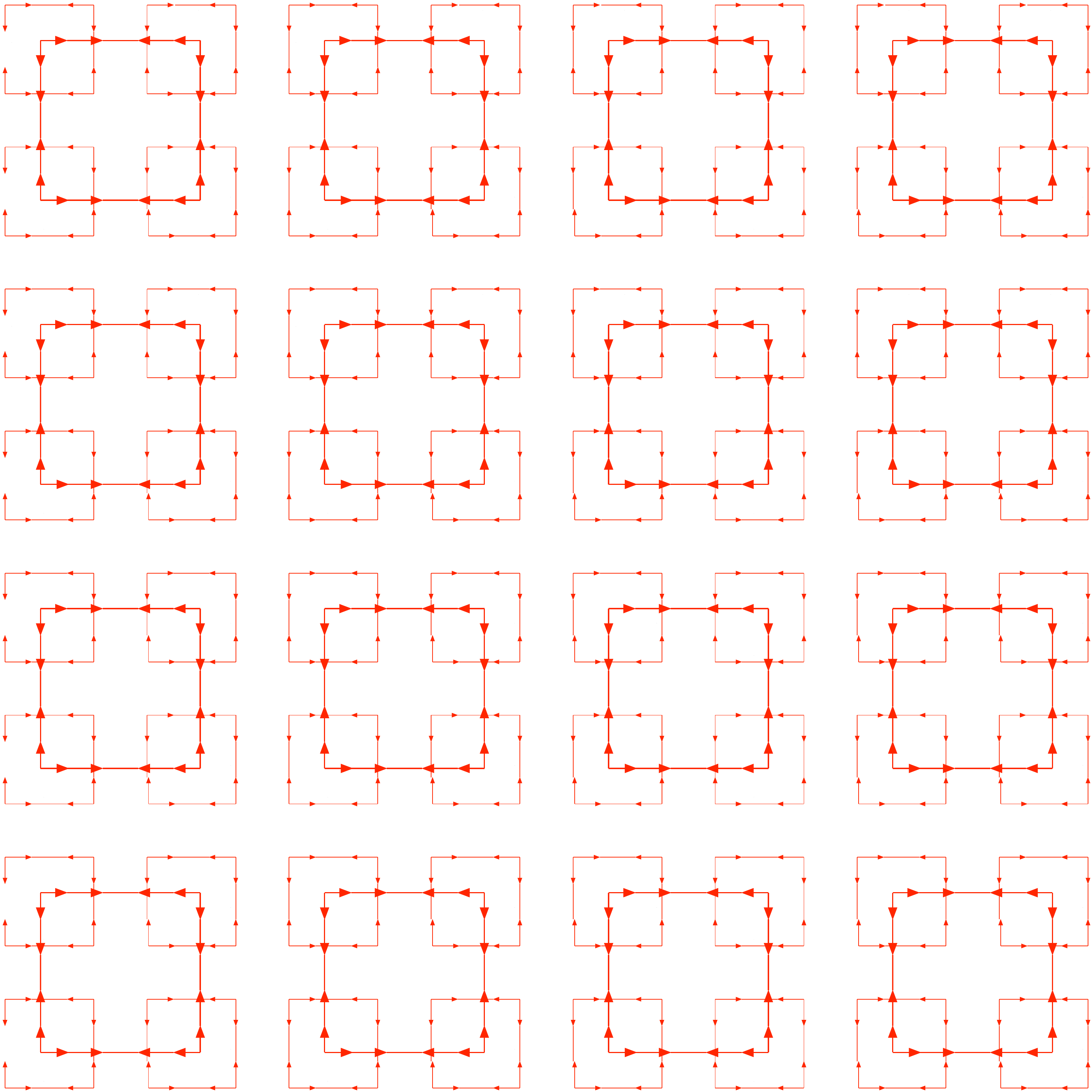}\label{fig:f2}}
  \caption{Hierarchical structure of Robinson tilings.}
\end{figure}

\begin{figure}[htp]
 \label{hier-b}
 \begin{center}
  \includegraphics[scale=0.4]{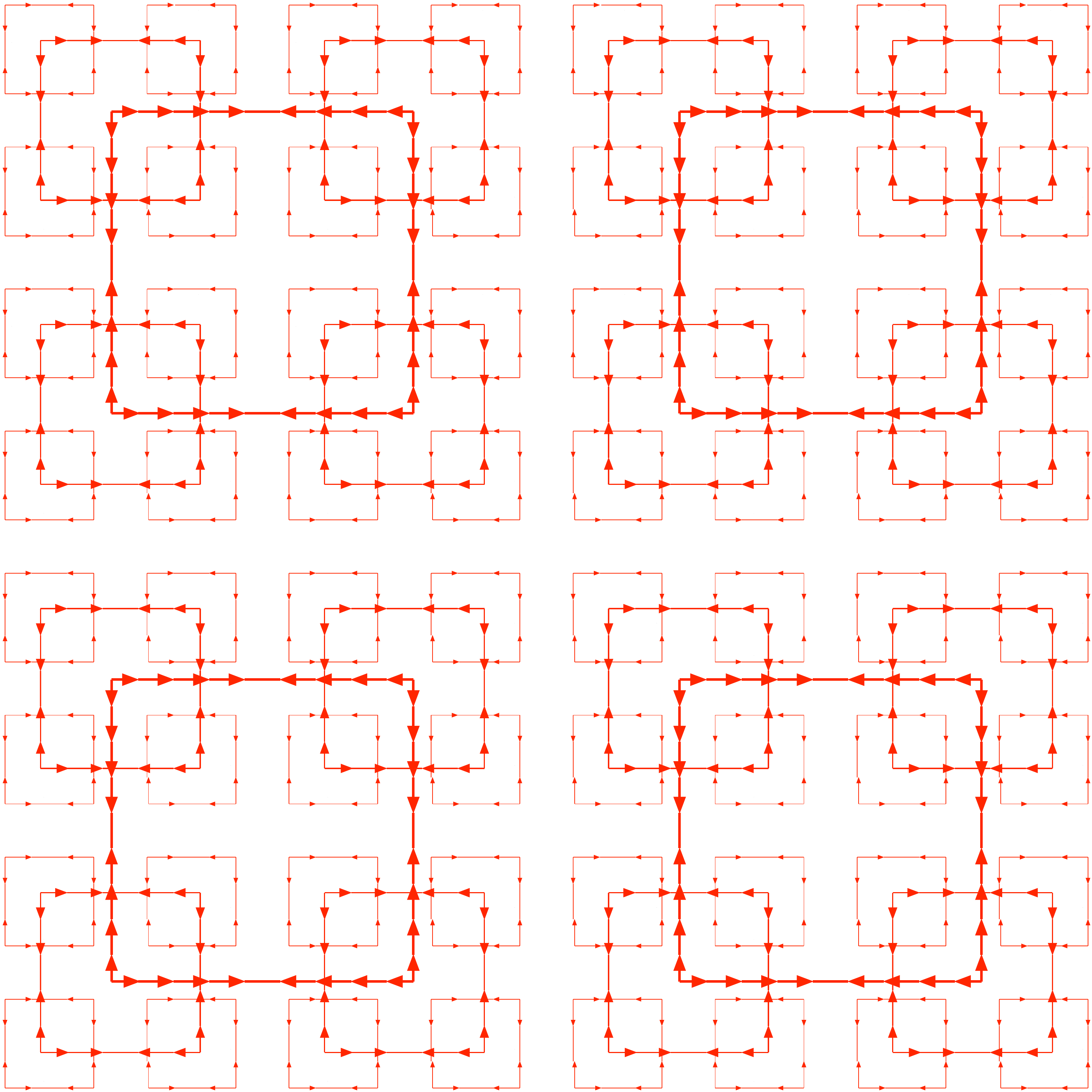}
 \end{center}
  \caption{Hierarchy.}
\end{figure}


We have an $n \times n $ square denoted by $S_n$. We will mark all cells of $S_n$ with a symbol $V$ if it is possible to place a corner tile in this cell, with a symbol $C$ if we have chosen this cell to be a corner and with a black dot if it is impossible to place a corner tile.  
Initially, we have just an $n \times n$ block with all cells marked with  $V$. The first step will be to choose places for the corner tiles of rank 1. There are only four variants to do so, example for step 1, $n = 5$, is in Figure \ref{step-one}. 
Then, at each step $i \ge1$, we choose places for corner tiles corresponding to supertiles of rank $i$. Let us remark that when we place a corner tile of rank $i > 1$, it fixes the orientation of the corner tiles of rank $i-1$ as in Figure \ref{step-two}. We stop when the number of vacant places is less than four. 
 \begin{figure}[htp]
 \begin{center}
  \includegraphics[scale=0.5]{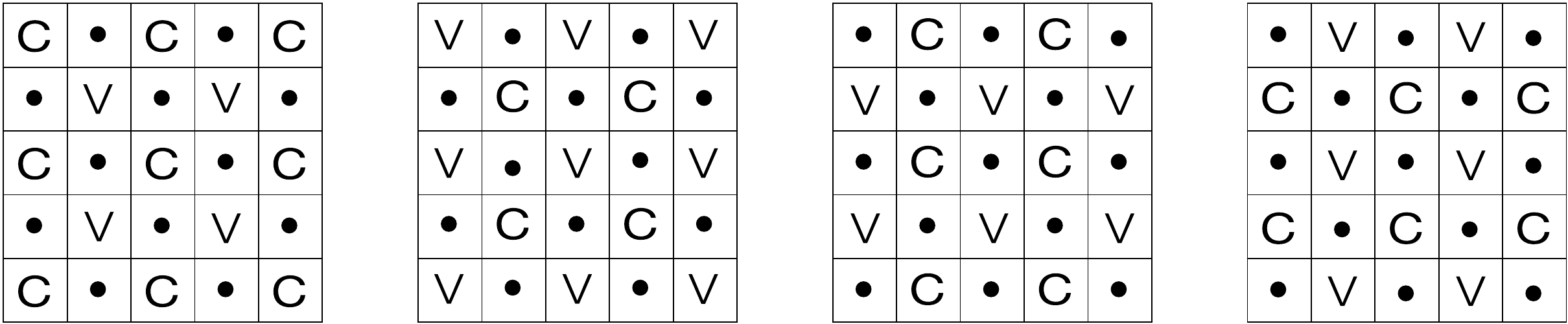}
 \end{center}
   \caption{Four variants for corner tiles at step one.}

\label{step-one}
\end{figure}

 \begin{figure}[htp]
 \begin{center}
  \includegraphics[scale=0.5]{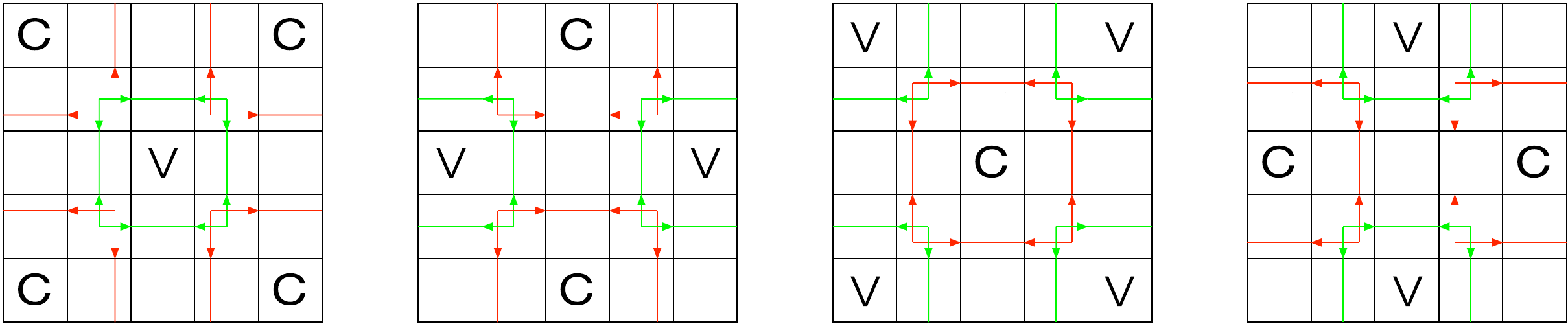}
 \end{center}
 \caption{Four variants for corner tiles at step two with a corner tile in position [1,2]. }

\label{step-two}
\end{figure}

 At each step, the set of all possible correct tilings of $S_n$ is divided into four disjoint classes. Resulting classes by construction can have either one or two vacant places. 

\begin{figure}[!tbp]
  \centering
  \subfloat[]{\includegraphics[width=0.45\textwidth]{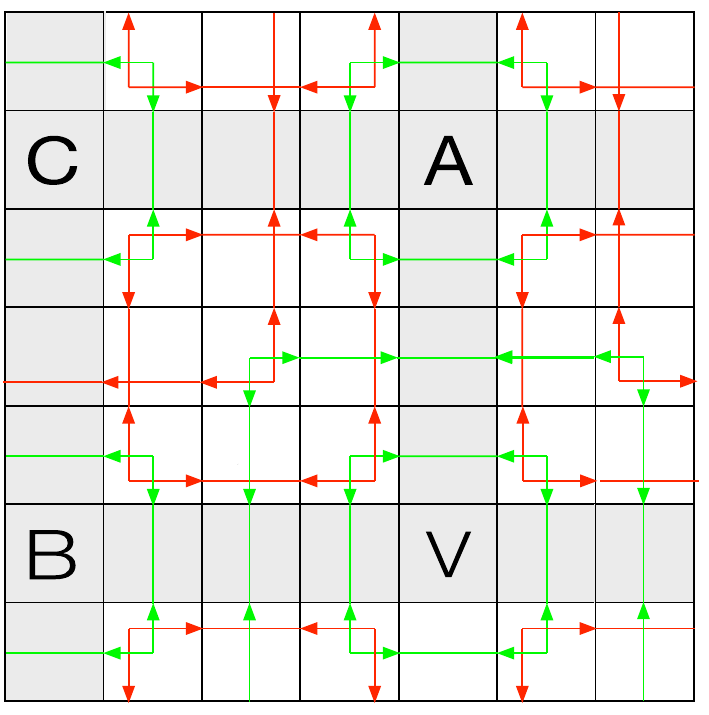}\label{fig:f1}}
  \hfill
  \subfloat[]{\includegraphics[width=0.45\textwidth]{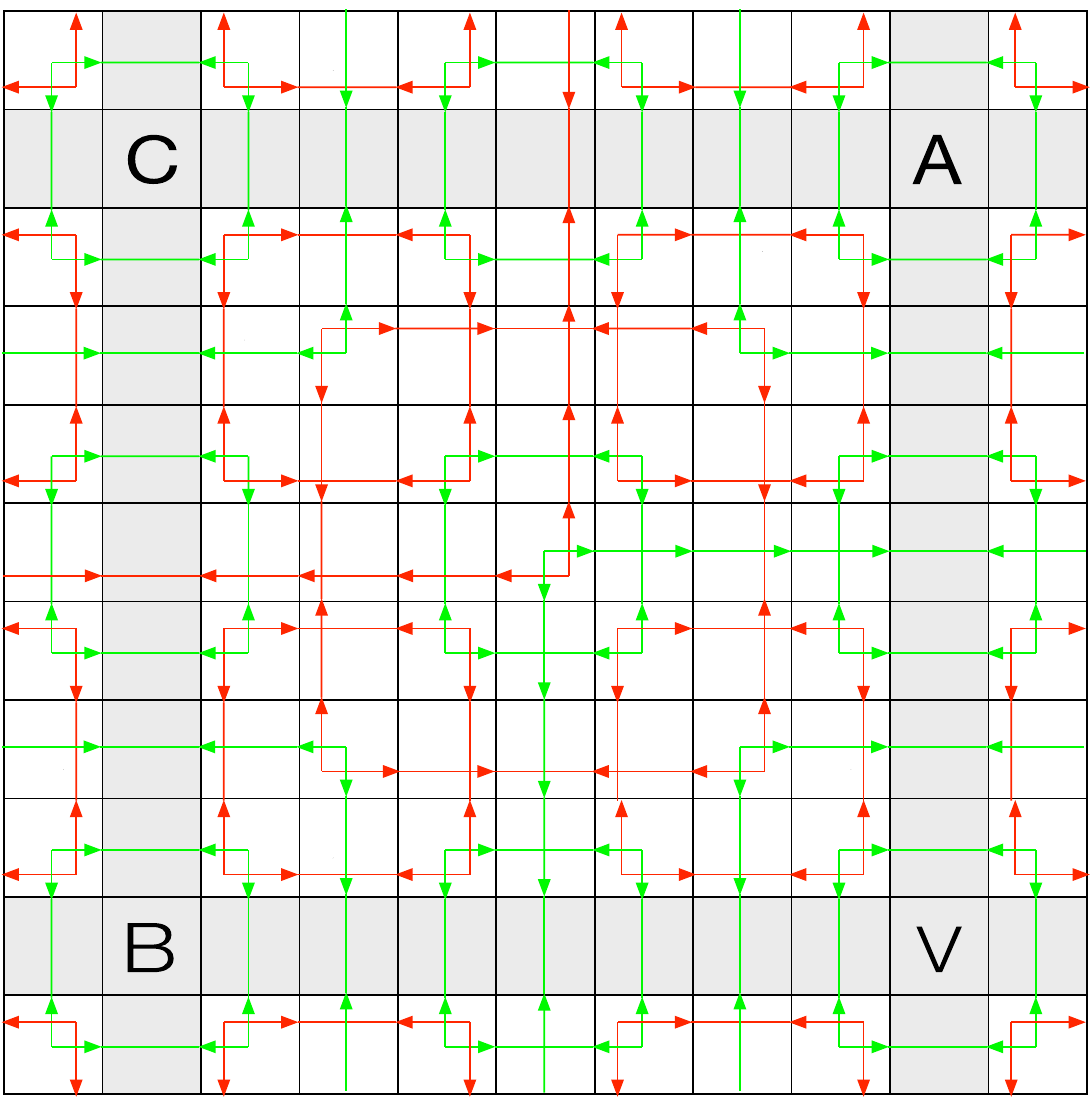}\label{fig:f2}}
  \caption{Example for one vacant place. All the cells are defined except for two rows and two columns. One of the 'crossroads' has to be a corner tile (marked with $C$) and we will mark the other two with letters $A$ and $B$.} 
  \label{1w}
\end{figure}

\begin{lemma}
\label{lem}
If two squares have the same number of vacant places for corner tiles, there is the same number of possibilities to complete them to a correct pattern.
\end{lemma}

\begin{proof}
 
 First let us prove that once $n,m>2$, if two squares $S_n$ and $S_m$ at the end of the procedure described earlier remain with the same number of vacant places (thus one of two), then there are exactly the same number of possibilities to complete both.

Suppose both squares have only one vacant place, then they have all the cells filled except for two rows and two columns. One of the 'crossroads' has to be a corner tile (marked with $C$) and the other two we will mark with letters $A$ and $B$ as in Figure \ref{1w}. Indeed, for any correct tiling of one square we can specify one tiling of the other simply by choosing the same tiles for cells marked with A, B, C, and V. Same holds true for the situation with two vacant places. 

To conclude the proof, it remains to note that the number of classes with one (two) vacant places is completely defined by the initial number of vacant places $V$.
\end{proof}

Using Lemma 1 we can obtain recurrence relations for the number of correct tilings. All we need to do is to count the number of vacant places after placing the corner tiles of the first and second ranks and then to find any smaller square with the same number of vacant places. 

Denote by $A_1$ the number of correct tilings of $S_2$ \emph{with} a corner tile in the position [1,1] and by $A_n$ with $n>1$ for the number of all correct tilings of $S_n$. Denote by $B_n$ the number of correct tilings of $S_{2n+1}$ with a corner tile in the position [1,2]. If $n>1$, $A_n$ is a sum of all possible tilings with corner tile in positions [1,1], [2,2], [2,1] and [1,2].

If $n = 2k$ is an even number then after the first step we have exactly $k \times  k$ vacant places for all four possibilities to place first rank corner tiles. By using Lemma 1 we can state that once $n>1$, the number of all correct tilings of $S_{2n}$ with a corner tile in the position [1,1] is equal to the number of all correct tiling of $S_n$. The same is true for any other choice at step one,  which gives us the first recurrence relation:
\begin{align}
A_{2n} \quad & = 4 \cdot  A_n, \qquad n>0 .
\end{align}

In the case when $ n  =  2k+1$  is odd number we can write the number of vacant places as follows :

\begin{itemize}
\item corner tile in the position [1,1] : $k \times k$  vacant places  ;
\item corner tile in the position [2,2] : $(k+1) \times (k+1)$   ;
\item corner tile in the position [1,2] : $k \times (k+1)$  ;
\item corner tile in the position [2,1] : $k \times (k+1)$  .
\end{itemize}

The number of possibilities for last the two options was already denoted by $B_n$. Now we can write the second recurrence relation:

\begin{align}
A_{2n+1}& = A_n + A_{n+1}+ 2\cdot B_n, \qquad n>0 
\end{align}

Now we need to find recurrence relations for $B_n$. Again, by carefully examining all the possibilities for the second order corner tiles we can write the last two recurrence relations: 

\begin{align}
B_{2n}  \quad& = 2 \cdot A_n  + 2 \cdot B_n ,\\
B_{2n+1} & = 2 \cdot A_{n+1} + 2 \cdot B_{n}, \qquad n>0.
\end{align}
Values of $A_1$ and $B_1$ can be found by an exhaustive search:

\begin{IEEEeqnarray}{rCl}
A_1 & = & 56 ;  \nonumber\\ \nonumber
B_1 & = & 124.
\end{IEEEeqnarray}
The solution for recurrence relations above can be written as  

\begin{equation} 
A_n = a(n) \cdot A_1 + b(n) \cdot B_1, \nonumber
\end{equation}
where 

\begin{IEEEeqnarray}{cccl}
\label{an}
a(n) & =  &  \> 5n^{2}  - 12 n \cdot 2^{\lfloor\log_2 n\rfloor} + 8 \cdot 2^{2\lfloor\log_2 n\rfloor}    \\
\label{bn}
b(n) & =  &-   2 n ^2 + 6 n \cdot  2^{\lfloor\log_2 n\rfloor}  - 4 \cdot 2^{2\lfloor\log_2 n\rfloor} 
\end{IEEEeqnarray}
The sum of (\ref{an}) and (\ref{bn}) gives us:

\begin{thm}
\label{T}
For any Robinson tiling made of one infinite order supertile, once $n>1$, the number of distinct $n \times n$ square blocks is given by 

\begin{IEEEeqnarray}{cccl}
A(n) & =  &  \> 32n^{2}  + 72 n \cdot 2^{\lfloor\log_2 n\rfloor} - 48 \cdot 2^{2\lfloor\log_2 n\rfloor}.
\nonumber
\end{IEEEeqnarray}
\end{thm}

\section{Concluding Remarks}

Franz Gähler and Johan Nilsson conjectured in [\ref{nilsson}] that the number of distinct $n \times n$ blocks in 2D paper folding sequence is equal to

\[
\label{paper}
P(n)=  12n^{2}  + 24 n \cdot 2^{\lfloor\log_2 n\rfloor} - 16 \cdot 2^{2\lfloor\log_2 n\rfloor} -4, \qquad n\ge3.
\]

This formula looks very similar to (\ref{comp}).  The author believes that the method described in this paper may be  modified to prove this conjecture as well as to find other complexity formulas, in particular, for other tilings that have Toeplitz-like structure (see [\ref{allouche}]).

\end{document}